\documentclass[12 points]{article}
\usepackage{amssymb}
\usepackage{eucal}
\usepackage{amsmath}
\usepackage{amsthm}

\newcommand{\R}{{\mathbb  R}}

\numberwithin{equation}{section}
\newtheorem{thm}{\bf Theorem}[section]
\newtheorem{lem}[thm]{\bf Lemma}
\newtheorem{prop}[thm]{\bf Proposition}
\newtheorem{cor}[thm]{\bf Corollary}
\newtheorem{defn}{\bf Definition}[section]

\theoremstyle{remark}
\newtheorem{rem}{\bf Remark}[section]

\begin{document}

\title{\large\bf DYNAMICS UNDER GEOMETRIC  DISSIPATION}
\author{Petre Birtea and Dan Com\u{a}nescu}
\date{ }
\maketitle

\begin{abstract}
We give sufficient conditions for asymptotic stabilization of equilibrium points and periodic orbits of a dynamical system when we add a geometric dissipation of gradient type. We also describe the domain of attraction in the case of asymptotic stability.

\end{abstract}



\section{Introduction}

We consider the dynamical system
\begin{equation}\label{unperturbed}
    \dot{x}=X(x),
\end{equation}
where $X\in \mathcal{X}(M)$ with $(M,g)$ a smooth finite dimensional Riemannian manifold. We will denote by $x_{un}(\cdot, x_0)$ the solution of \eqref{unperturbed} with the initial condition $x_0$.
Suppose that we have $F_1,...F_k,G\in C^{\infty}(M)$  conserved quantities for dynamics \eqref{unperturbed}. In \cite{birtea-comanescu-geometric} have been constructed a perturbation vector field that conserves $F_1,...,F_k$ and dissipates $G$ after a prescribed rule given by $h\in C^{\infty}(M)$. If we take $h(x)=\det\Sigma_{(F_1,...,F_k,G)}^{(F_1,...,F_k,G)}(x)$, then the perturbation is given by the standard control vector field
\begin{equation}\label{v0}
    \mathbf{v_0}=\sum_{i=1}^k(-1)^{i+k+1}\det \Sigma_{(F_1,...,\widehat{F_i},...,F_k,G)}^{(F_1,...,F_k)}\nabla
    F_i+\det\Sigma_{(F_1,...,F_k)}^{(F_1,...,F_k)}\nabla G.
\end{equation}
For $f_1,...,f_r,g_1,...,g_s:M\rightarrow \mathbb{R}$ smooth functions on the manifold $(M,g)$ and $<\cdot,\cdot>$ the scalar product induced by Riemannian metric $g$, we use the notation
\begin{equation}\label{sigma}
\Sigma_{(g_1,...,g_s)}^{(f_1,...,f_r)}=\left(%
\begin{array}{cccc}
  <\nabla g_1,\nabla f_{1}> & ... & <\nabla g_s,\nabla f_{1}> \\
  ... & ... & ... \\

  <\nabla g_1,\nabla f_r> & ... & <\nabla g_s,\nabla f_r> \\
\end{array}%
\right).
\end{equation}
In \cite{birtea-comanescu-geometric} has been given three different formulations for the standard control vector field $\mathbf{v_0}$:
\begin{itemize}
\item [(i)] the covariant formulation,
$\mathbf{v}_0=(-1)^{n+1}\sharp_g(*(dF_1\wedge ... \wedge dF_k\wedge *(dG\wedge dF_1\wedge ... \wedge dF_k)));$
\item [(ii)] the contravariant formulation,
$\mathbf{v}_0=\mathbf{i}_{dG}\mathbf{T}$, where $\mathbf{T}:\Omega^1(M)\times\Omega^1(M)\rightarrow \mathbb{R}$ is the symmetric contravariant 2-tensor given by
\begin{equation}\label{T}
\mathbf{T}:=\sum_{i,j=1}^k(-1)^{i+j+1}\det\Sigma_{(F_1,...,\hat{F_i},...,F_k)}^{(F_1,...,\hat{F_j},...,F_k)}\nabla F_i\otimes\nabla F_j+\det\Sigma_{(F_1,...,F_k)}^{(F_1,...,F_k)} g^{-1};
\end{equation}
\item [(iii)] the formulation with orthogonal projection,
$\mathbf{v}_0(x)=\det\Sigma_{(F_1,...,F_k)}^{(F_1,...,F_k)}(x) P_{T_xL_c} (\nabla G(x)),$ where $L_c$ is the regular leaf of $\mathbf{F}=(F_1,...,F_k):M\rightarrow \mathbb{R}^k$ which contains $x$ and $P_{T_xL_c}:T_xM\rightarrow T_xM$ is the orthogonal projection.
\end{itemize}

The aim of this paper is to study the dynamics of {\bf the geometrically dissipated system}
\begin{equation}\label{perturbed}
    \dot{x}=X(x)-\mathbf{v}_0(x).
\end{equation}
We will denote by $x_{p}(\cdot, x_0)$ the solution of \eqref{perturbed} with the initial condition $x_0$. By construction of the standard control vector field ${\bf v}_0$ we have that the function $G$ decreases along the solutions of the geometrically dissipated system  \eqref{perturbed}, i.e. 
\begin{equation}\label{G-crescator}
\frac{dG}{dt}(x_p(t,x))=-\det\Sigma_{(F_1,...,F_k,G)}^{(F_1,...,F_k,G)}(x_p(t,x))\leq 0.
\end{equation}

The standard control vector field can be formally written as
\begin{equation*}
\mathbf{v}_0(x)=\det \left(%
\begin{array}{cccc}
  <\nabla F_1(x),\nabla F_{1}(x)> & ... & <\nabla F_k(x),\nabla F_{1}(x)> & <\nabla G(x),\nabla F_1(x)> \\
  ... & ... & ... & ... \\
  <\nabla F_1(x),\nabla F_{k}(x)> & ... & <\nabla F_k(x),\nabla F_{k}(x)> & <\nabla G(x),\nabla F_k(x)>\\
\nabla F_1(x) & ... & \nabla F_k(x) & \nabla G(x) \\
\end{array}%
\right)
\end{equation*}
and consequently, we have the following result.

\begin{lem}\label{eq-v0}
For $x\in M$ the following are equivalent:
\begin{itemize}
\item [(i)] $\mathbf{v}_0(x)=0$; 

\item [(ii)] $\nabla F_1(x),...,\nabla F_k(x),\nabla G(x)$ are linear dependent;

\item [(iii)] $\det\Sigma_{(F_1,...,F_k,G)}^{(F_1,...,F_k,G)}(x)=0$.
\end{itemize}
\end{lem}

\begin{proof}
For the implication $(ii)\Rightarrow (i)$ we have that in the formal determinant that defines ${\bf v}_0$ one of the column is a linear combination of the remaning columns. 
The implication  $(i)\Rightarrow (ii)$ is obvious from the definition of ${\bf v}_0$. 
The equivalence between  $(ii)$ and  $(iii)$ is a well known result in linear algebra.

\end{proof}

\section{Equilibrium points}

In this section we study the equilibrium points for the geometrically dissipated system \eqref{perturbed}.

\begin{prop}\label{eq-equilibriums}
We have $X(x)-\mathbf{v}_0(x)=0$ if and only if $X(x)=0$ {\bf and} $\mathbf{v}_0(x)=0$.
\end{prop}

\begin{proof}
The implication $"\Leftarrow"$ is trivial. 

For the other implication, from \eqref{G-crescator} we have that $\frac{dG}{dt}(x_p(t,x))=-\det\Sigma_{(F_1,...,F_k,G)}^{(F_1,...,F_k,G)}(x_p(t,x))$. If $x$ is an equilibrium point of the vector field $X-\mathbf{v}_0$ then $\det\Sigma_{(F_1,...,F_k,G)}^{(F_1,...,F_k,G)}(x)=0$ and consequently, $\mathbf{v}_0(x)=0$ and also $X(x)=0$.
\end{proof}

We denote by $E_{un}$ and $E_p$ the sets of equilibrium points for the unperturbed system \eqref{unperturbed}, respectively the geometrically dissipated system \eqref{perturbed}.
A relevant set for the perturbed dynamics is given by
\begin{equation}\label{inv}
    Inv:=\{x\in M\,|\,\det\Sigma_{(F_1,...,F_k,G)}^{(F_1,...,F_k,G)}(x)=0\}=\{x\in M\,|\,{\bf v}_0(x)=0\}.
\end{equation}

Using Lemma \ref{eq-v0} and Proposition \ref{eq-equilibriums} we obtain the following characterization of equilibrium points for the perturbed system.
\begin{thm}\label{equilibria}
 The set of equilibria for the geometrically dissipated system is characterized by the equality $$E_p= E_{un}\cap Inv.$$

\end{thm}

By perturbing the initial dynamics  \eqref{unperturbed}  with the standard control vector field $\mathbf{v}_0$, some of unperturbed equilibrium points will not remain equilibrium points for the geometrically dissipated system. We loose exactly that equilibrium points $x_e$ for which the vectors $\nabla F_1(x_e)$,...,$\nabla F_k(x_e)$ and $\nabla G(x_e)$ are linear independent.
The set $Inv$ is invariant under the unperturbed dynamics  \eqref{unperturbed} (see Corollary 2.4 in \cite{birtea-comanescu}). 
Next we will prove that the set $Inv$ is also an invariant set for the geometrically dissipated dynamics \eqref{perturbed}. Moreover, a solution of the unperturbed system which start from $Inv$ is also a solution for the geometrically dissipated system.

\begin{thm}\label{invarianta}
We have the following properties:
\begin{itemize}
\item [(i)] For an initial condition $x\in Inv$ we have $x_{un}(t,x)=x_p(t,x)$.
\item [(ii)] The set $Inv$ is invariant under the geometrically dissipated system \eqref{perturbed}.
\item [(iii)] If $x\notin Inv$, then $\frac{dG}{dt}(x_p(t,x))<0$ for all $t$.

\end{itemize}
\end{thm}

\begin{proof}
{\it (i)} Let $x\in Inv$ and $x_{un}(t,x)$ be the solution of \eqref{unperturbed} starting from initial condition $x$. From the invariance of $Inv$ under the unperturbed dynamics, we obtain that $x_{un}(t,x)\in Inv$ for all $t$. By Lemma \ref{eq-v0}, we have that $\mathbf{v}_0(x_{un}(t,x))=0$ for all $t$. This shows that $x_{un}(t,x)$ is also a solution of the geometrically dissipated system \eqref{perturbed}.

\noindent {\it (ii)} It is an immediate consequence of {\it (i)}.

\noindent {\it (iii)} From \eqref{G-crescator} we have that $\frac{dG}{dt}(x_p(t,x))=-\det\Sigma_{(F_1,...,F_k,G)}^{(F_1,...,F_k,G)}(x_p(t,x))\leq 0$. If $x\notin Inv$, then from {\it (ii)} we have that $x_p(t,x)\notin Inv$, for all $t$. By the definition of $Inv$, we obtain the strict inequality.
\end{proof}

An immediate consequence of Theorem \ref{invarianta} {\it (i)} and {\it (iii)} is the following result.

\begin{cor}
If the geometrically dissipated system \eqref{perturbed} has a periodic orbit or a homoclinic orbit or a heteroclinic cycle, they are contained in $Inv$ and they are also a periodic orbit, respectively homoclinic orbit or heteroclinic cycle for the unperturbed system \eqref{unperturbed}.
\end{cor}

For the invariant set $Inv$ we have a further decomposition
$$Inv=G_*\cup Y,$$
where $G_*:=\{x\in M\,|\,\nabla G(x)=0\}$ and $Y$ is the complementary set of $G_*$ in $Inv$. The subsets $G_*$ and $Y$ are also invariant subsets for the perturbed dynamics. Indeed, $G_*$ is invariant under the unperturbed dynamic (see \cite{birtea-comanescu}, \cite{irtegov-titorenko}) and if $x\in G_*$ is an initial condition then $x_p(t,x)=x_{un}(t,x)\in G_*$ for all $t\in \mathbb{R}$, which implies that $G_*$ is invariant under the perturbed dynamics. Consequently, $Y$ is also an invariant set for the perturbed dynamics. By construction we have $$Y=\{x\in M\,|\,\nabla G(x)\neq 0\,\,\text{and}\,\,\nabla F_1(x),...,\nabla F_k(x),\nabla G(x)\,\,\text{are linear dependent}\}.$$

\section{Perturbed dynamics on the regular leaves}

 In this section we will study the geometrically dissipated dynamics \eqref{perturbed} restricted to a regular leaf $L_c:=\mathbf{F}^{-1}(c)$ generated by a regular value of the function $\mathbf{F}:=(F_1,...,F_k):M\rightarrow \mathbb{R}^k$. Every regular leaf $L_c$ is invariant under perturbed dynamics \eqref{perturbed} as both vector fields $X$ and $\mathbf{v}_0$ are tangent vector fields to the leaves.

Next we will give a characterization of the invariant set $Inv\cap L_c$ for the perturbed dynamics restricted to the regular leaf $L_c$.
It has been proved in \cite{birtea-comanescu-geometric}, Theorem 4.5. that $\mathbf{v_0}_{_{|L_c}}=\nabla_{\tau_c}G_{|L_c}$, where $\tau_c=\frac{1}{\det\Sigma_{(F_1,...,F_k)}^{(F_1,...,F_k)}\circ i_c}i_c^*g$ is a conformal metric with the induced metric $i_c^*g$ on $L_c$. Consequently, we have
$$Inv\cap L_c=\{x\in L_c\,|\,\nabla_{\tau_c}G_{|L_c}(x)=0\}.$$
The derivative of the function $G_{|L_c}$ along the solution of the geometrically dissipated dynamics \eqref{perturbed} restricted to the regular leaf $L_c$ is given by
\begin{eqnarray*}
\dot{G}_{|L_c}(x)&=&L_{(X_{_{|L_c}}-\mathbf{v_0}_{_{|L_c}})}G_{|L_c}(x)=\tau_c(x)(X_{_{|L_c}}(x)-
\mathbf{v_0}_{_{|L_c}}(x),\nabla_{\tau_c}G_{|L_c}(x))\nonumber \\
&=& \tau_c(x)(X_{_{|L_c}}(x),\nabla_{\tau_c}G_{|L_c}(x))-\tau_c(x)(\mathbf{v_0}_{_{|L_c}}(x),\nabla_{\tau_c}G_{|L_c}(x))\nonumber \\
&=& L_{X_{_{|L_c}}}G_{|L_c}(x)-\tau_c(x)(\nabla_{\tau_c}G_{|L_c}(x),\nabla_{\tau_c}G_{|L_c}(x)) \nonumber \\
&=& -||\nabla_{\tau_c}G_{|L_c}(x)||_{\tau_c}^2,
\end{eqnarray*}
for any $x\in L_c$. In the above computations we have used the fact that the function $G$ is a conserved quantity for the vector field $X$ which implies $L_{X_{_{|L_c}}}G_{|L_c}=0$. Consequently, we obtain the set equality
\begin{equation}\label{G-punct}
    Inv\cap L_c=\{x\in L_c\,|\,\dot{G}_{|L_c}(x)=0\}.
\end{equation}

\begin{rem}
The set $\{x\in L_c\,|\,\dot{G}_{|L_c}(x)=0\}$ is the key set that appears in LaSalle Invariance Principle. The largest invariant set contained in $\{x\in L_c\,|\,\dot{G}_{|L_c}(x)=0\}$ is the set that contains informations about asymptotic behaviour of certain solutions. For our case the the largest invariant set contained in $\{x\in L_c\,|\,\dot{G}_{|L_c}(x)=0\}$ is the set itself as being equal with $Inv\cap L_c$.
\end{rem}
\medskip

In what follows we study the asymptotic behaviour of the solutions for the geometrically dissipated system \eqref{perturbed} restricted to a regular leaf $L_c$. 
We suppose that any solutions of the geometrically dissipated system \eqref{perturbed} are defined on $\mathbb{R}$.
The $\omega$-limit set of $x_0$ is
$$\omega(x_0):=\{z\in L_c\,|\,\exists t_1, t_2...\rightarrow \infty\,\,\texttt{s.t.}\,\,x_p(t_k, x_0)\rightarrow z \,\,\texttt{as}\,\,k\rightarrow\infty\}$$
The $\omega$-limit sets have the following properties that we will use later. For more details, see \cite{robinson}.

\begin{itemize}
  \item [(i)] If $x_p(t,y)=z$ for some $t\in \R$, then
  $\omega(y)=\omega(z)$.
  \item [(ii)] $\omega(x_0)$ is a closed subset and both positively and negatively
  invariant (contains complete orbits).
\end{itemize}

We have the following LaSalle type result.

\begin{thm}\label{omega-limit} (Invariance Principle for the geometrically dissipated system)

Let $x_0$ be an arbitrary point in $L_c$, then the following holds:
\begin{itemize}
  \item [(i)] If $a,b\in \omega (x_0)$ then $G(a)=G(b)\leq G(x_0)$. Equality holds if and only if $x_0\in Inv\cap L_c$.
  \item [(ii)] We have the following set inclusion $\omega(x_0)\subset Inv\cap L_c$.

  \item [(iii)] If $\{x_p(t,x_0)\,|\,t\geq 0\}$ is bounded then $\omega (x_0)$ is compact and nonempty and moreover $$\lim_{t\rightarrow \infty}d_{\tau_c}(x_p(t,x_0),Inv\cap L_c)=0,$$
where $d_{\tau_c}$ is the distance function on $L_c$ induced by the Riemannian metric $\tau_c$.
\end{itemize}
\end{thm}

\begin{proof}
For $(i)$, let $a,b\in \omega(x_0)$. There exists two sequences $(t_n)_{n\in \mathbb{N}}$ and $(s_n)_{n\in \mathbb{N}}$ such that $t_n<s_n<t_{n+1}$ with $t_n\rightarrow \infty$ and $x_p(t_n,x_0)\rightarrow a$, $x_p(s_n,x_0)\rightarrow b$. Consequently, as $G$ is a decreasing function along the solution $x_p(\cdot ,x_0)$ we have the inequality $G(x_p(t_{n+1} ,x_0))\leq G(x_p(s_n ,x_0))\leq G(x_p(t_{n} ,x_0))$. Taking the limit we obtain $G(a)=G(b)\leq G(x_0)$. If $x_0\in Inv$ then $x_p(t,x_0)\in Inv$ and $G(x_p(t,x_0))=G(x_0)$ for all $t\in\mathbb{R}$. Reciprocally, if $G(a)=G(b)= G(x_0)$ then using Theorem \ref{invarianta} $(iii)$ we obtain the enounced result.

$(ii)$ Let $a\in \omega(x_0)$, then $x_p(t,a)\in \omega (x_0)$ for all $t$. From $(i)$ we have that $G(x_p(t,a))=G(a)$ for all $t$ and  consequently, $\frac{dG}{dt}(x_p(t,a))=0$. We obtain that $x_p(t,a)\in Inv$ for all $t$ and in particular $a\in Inv$.

$(iii)$ The first part is a classical result, see \cite{bhatia-szego}, \cite{robinson}. We have $\omega(x_0)\subset Inv\cap L_c$ and consequently, $d_{\tau_c}(x_p(t,x_0),Inv\cap L_c)\leq d_{\tau_c}(x_p(t,x_0),\omega(x_0))$. But $\lim_{t\rightarrow \infty}d_{\tau_c}(x_p(t,x_0),\omega(x_0))=0$. 
\end{proof}

Our next purpose is to study the change of stability for equilibrium points of the unperturbed dynamics restricted to a regular leaf $L_c$ when we add the geometric dissipation of gradient type $-{\bf v}_0$. Because the added dissipation $-{\bf v}_0$ is of gradient type when restricted to a regular leaf $L_c$, it is to be expected that the stability of an equilibrium point for the perturbed system to be dictated by the nature of the equilibrium point as a critical point for the function $G_{|L_c}$.

\begin{thm} Let $x_e\in L_c$ be a locally strict minimum for $G_{|L_c}$. Then the following holds
\begin{itemize}
\item [(i)] $x_e$ is an asymptotically stable equilibrium for the geometrically dissipated system \eqref{perturbed} restricted to $L_c$.
\item [(ii)] There exists $k>G(x_e)$ such that $cc_{x_e}G_{|L_c}^{-1}([G(x_e), k])\cap Inv=\{x_e\}$. (The set $cc_{x_e}G_{|L_c}^{-1}([G(x_e), k])$ is the connected component of $G_{|L_c}^{-1}([G(x_e), k])$  that contains the point $x_e$.)
\item [(iii)] If $G_{|L_c}:L_c\rightarrow \mathbb{R}$ is a proper function, then
for any $k>G(x_e)$ for which $cc_{x_e}G_{|L_c}^{-1}([G(x_e),k])\cap Inv=\{x_e\}$ the set  $cc_{x_e}G_{|L_c}^{-1}([G(x_e),k])$ is included in the domain of attraction of the asymptotically stable equilibrium point $x_e$.
\end{itemize}
\end{thm}

\begin{proof}
$(i)$ We prove that if $x_e$ is a locally strict minimum for $G_{|L_c}$ then $x_e$ is isolated in $Inv\cap L_c$. Indeed, we have $x_e\in cc_{x_e}\{x\in L_c\,|\,\nabla_{\tau_c}G_{|L_c}(x)=0\}=cc_{x_e}(Inv\cap L_c)$. Using Sard Theorem, see \cite{sternberg} and \cite{arsie-ebenbauer},  Lemma 10, we have the inclusion $cc_{x_e}\{x\in L_c\,|\,\nabla_{\tau_c}G_{|L_c}(x)=0\}\subset G_{|L_c}^{-1}(G(x_e))$. But $x_e$ is a locally strict minimum for $G_{|L_c}$ and consequently, it is isolated in $G_{|L_c}^{-1}(G(x_e))$ which also implies that it is isolated in $Inv\cap L_c$.

From the invariance of $Inv\cap L_c$ for the perturbed dynamics \eqref{perturbed} restricted to $L_c$ and the fact that $x_e$ is isolated in $Inv\cap L_c$ we obtain that $x_e$ is an equilibrium point. From locally strict minimality of $x_e$ and the fact that $x_e$ is isolated in $Inv\cap L_c$ there exists a small neighborhood $U$ in $L_c$ of $x_e$ such that $G_{|L_c}(x)-G_{|L_c}(x_e)>0$ and  $\dot{G}_{|L_c}(x)<0$ for any $x\in U\backslash \{x_e\}$. By Lyapunov theorem we obtain that $x_e$ is an asymptotically stable equilibrium for the perturbed dynamics \eqref{perturbed} restricted to $L_c$.
\medskip

$(ii)$ We have shown that $x_e$ is isolated in the set $Inv\cap L_c$ and consequently, there exists a closed ball $\overline{B}(x_e,r)$ with $G_{|L_c}(x)>G_{|L_c}(x_e),\,\forall x\in (\overline{B}(x_e,r)\cap L_c)\backslash\{x_e\}$ and $\overline{B}(x_e,r)\cap Inv\cap L_c=\{x_e\}$. There exists $k_0=\min\limits_{x\in S(x_e,r)\cap L_c}G_{|L_c}(x)>G_{|L_c}(x_e)$, where $S(x_e,r)$ is the sphere with the radius $r$ and centered at $x_e$. If $k\in (G(x_e),k_0)$ then, $cc_{x_e}G_{|L_c}^{-1}([G(x_e),k])\subset \overline{B}(x_e,r)\cap L_c$. We will prove this by contradiction, indeed suppose there exists $y\in cc_{x_e}G_{|L_c}^{-1}([G(x_e),k])$ and $y\notin \overline{B}(x_e,r)\cap L_c$. There exists a continuous arc $a_{y,x_e}$ included in $cc_{x_e}G_{|L_c}^{-1}([G(x_e),k])$ connecting $y$ and $x_e$. Therefore, there exists $z\in a_{y,x_e}\cap S(x_e,r)\cap L_c$. Consequently, $k_0>k\geq G(z)$ which is a contradiction with the fact that $k_0$ is the minimum value of $G_{|L_c}$ for points in the sphere $S(x_e,r)\cap L_c$.
\medskip

$(iii)$ We prove that the set $D_{x_e}:=cc_{x_e}G_{|L_c}^{-1}([k,G(x_e)])$ is an invariant set for the perturbed dynamics \eqref{perturbed} under the hypothesis of $(iii)$.

We will proceed by contradiction, we suppose that $D_{x_e}$ is not invariant under \eqref{perturbed}. There exists $x_0\in D_{x_e}\backslash\{x_e\}$ and $t^*>0$ such that $G_{|L_c}(x_p(t^*,x_0))=G(x_e)$ and this is a consequence of the fact that $G_{|L_c}$ is an increasing function along the solutions of \eqref{perturbed}. As $x_e$ is a critical point for \eqref{perturbed} we have that $x_p(t^*,x_0)\neq x_e$. We have the following partition
$$D_{x_e}\cap G_{|L_c}^{-1}(G(x_e))=\{x_e\}\cup Y,$$
where $x_e\notin Y$, and $x_p(t^*,x_0)\in Y$ and $Y$ is a compact set in the relative topology of $D_{x_e}$. The compactness of $Y$ is a consequence of the the compactness of $D_{x_e}\cap G_{|L_c}^{-1}(G(x_e))$ and the fact that $x_e$ is isolated in $G_{|L_c}^{-1}(G(x_e))$ as being a locally strict maximum.

Because $D_{x_e}$ has $T_3$ separability property there exists two open neighborhoods $V_{x_e}$ and $V_Y$ (in the relative topology of $D_{x_e}$) of $x_e$ and respectively $Y$ such that $V_{x_e}\cap Y=\emptyset$.

Let $S:=D_{x_e}\backslash (V_{x_e}\cup Y)$. The set $S$ is a closed set in the compact set $D_{x_e}$ and consequently it is compact and by construction separates $x_e$ and $x_p(t^*,x_0)\in Y$. By the Mountain Pass Theorem (see \cite{jabri} ) there exists a point $x^*\in D_{x_e}$ which is a local maximum or a mountain pass point for $G_{|D_{x_e}}$ with $G(x^*)<G(x_e)$. According to Lemma \ref{mountain-anexe} (see Annexe), we have that $x^*$ is a local maximum or a mountain pass point for $G_{|L_c}$ restricted to the set $cc_{x_e}G^{-1}_{|L_c}((k-\varepsilon, G(x_e)])$, where $\varepsilon>0$ is small. Because $G(x^*)<G(x_e)$ we have that $x^*$ is a local maximum or a mountain pass point for $G_{|L_c}$ restricted to the open set $\overset{\circ}{\overbrace{cc_{x_e}G^{-1}_{|L_c}((k-\varepsilon, G(x_e)])}}$. Consequently, $x^*$ is a critical point for $G_{|L_c}$ which implies that $x^*\in Inv$. We have obtained a contradiction which shows that $D_{x_e}$ is a compact invariant set for the perturbed dynamics \eqref{perturbed}.

Let $x_0\in D_{x_e}$ be arbitrary. Because $D_{x_e}$ is invariant and compact we obtain that $\omega(x_0)\neq \emptyset$ and $\omega(x_0)\in D_{x_e}$. But also $\omega(x_0)\in Inv\cap L_c$ by Theorem \ref{omega-limit} $(ii)$. By hypothesis $\omega(x_0)=\{x_e\}$ and as $x_e$ is also asymptotically stable we obtain $\underset{t\rightarrow \infty}{\lim} x_p(t,x_0)=x_e$.
\end{proof}

We notice that the condition $x_e\in L_c$ being a locally strict minimum for $G_{|L_c}$ is equivalent with the following two conditions: $x_e\in L_c$ is a local minimum for $G_{|L_c}$ and $x_e$ is isolated in $Inv\cap L_c$.
We can summarize as follows:
{\bf \begin{itemize}
\item [(1)] Suppose $x_e$ is stable for the unperturbed dynamics \eqref{unperturbed} restricted to the leaf $L_c$. 
\begin{itemize}
\item [(1.1)] If $x_e$ is strict local minimum for $G_{|L_c}$ (which implies that it is isolated in $Inv\cap L_c$), then $x_e$ is an asymptotically stable equilibrium for the geometrically dissipated system \eqref{perturbed} restricted to the leaf $L_c$.

\item [(1.2)]  If $x_e$ is not a strict local minimum for $G_{|L_c}$ but it is still an isolated point in the set $Inv\cap L_c$, then $x_e$ is an unstable equilibrium for the geometrically dissipated system \eqref{perturbed} restricted to the leaf $L_c$.
\end{itemize}
 
\item [(2)] Suppose $x_e$ is unstable for the unperturbed dynamics \eqref{unperturbed} restricted to the leaf $L_c$, then $x_e$ remains an  unstable equilibrium for the geometrically dissipated system \eqref{perturbed} restricted to the leaf $L_c$.

\end{itemize}}

Also, if $x_e\in L_c$ is a locally strict extremum for $G_{|L_c}$ we obtain that $x_e$ is a stable equilibrium point for the unperturbed dynamics \eqref{unperturbed}. This is a consequence of the algebraic method for stability, see \cite{comanescu-1}, \cite{comanescu-2}, and \cite{comanescu-3}, i.e. $x_e$ is an isolated solution of the algebraic system 
$$F_1(x)=F_1(x_e),\,...\,,\,F_k(x)=F_k(x_e),\, G(x)=G(x_e).$$

The passage from asymptotic stability of equilibrium points of the  geometrically dissipated system \eqref{perturbed} to the asymptotic stability of periodic orbits for the  geometrically dissipated system \eqref{perturbed} is allowed by Theorem \ref{omega-limit}. More precisely, for periodic orbits we have the following stability result. 

\begin{thm}
Let $x_p(\R,x_0)$ be a periodic orbit for the  geometrically dissipated system \eqref{perturbed}. Assume that $x_p(\R,x_0)=cc_{x_0}(Inv \cap L_c)$ and all $y\in x_p(\R,x_0)$ are local minima for $G_{|L_c}$. Then the following holds:
\begin{itemize}
\item [(i)] The periodic orbit $x_p(\R,x_0)$ is asymptotically stable for the geometrically dissipated system \eqref{perturbed} restricted to $L_c$. (Asymptotic stability is understood in the sense of asymptotic stability of an invariant set of a dynamical system).
\item [(ii)] There exists $k>G(x_0)$ such that $cc_{x_0}G_{|L_c}^{-1}([G(x_0), k])\cap Inv=x_p(\R,x_0)$. 
\item [(iii)] If $G_{|L_c}:L_c\rightarrow \mathbb{R}$ is a proper function, then
for any $k>G(x_0)$ for which $cc_{x_0}G_{|L_c}^{-1}([G(x_0),k])\cap Inv=x_p(\R,x_0)$ the set  $cc_{x_0}G_{|L_c}^{-1}([G(x_0),k])$ is included in the domain of attraction of the asymptotically stable periodic orbit $x_p(\R,x_0)$.
\end{itemize}

\end{thm}

\section{Annexe}

The following results are taken from book \cite{jabri}.

\begin{defn}
Let $X$ be a topological space and $f : X \rightarrow \mathbb{R}$ be a continuous function. A point
$x\in X$ is called a mountain pass point (in the sense of Katriel) if for every neighborhood
$\mathcal{N}$ of $x$, the set
$$\mathcal{N}\cap \{y\in X\,|\, f (y) > f (x)\}$$
is disconnected.
\end{defn}

\begin{lem}\label{mountain-anexe}
Let $f:\widetilde{X}\rightarrow \mathbb{R}$ be a continuous function and let $X\subset \widetilde{X}$ be a subset with the property that
$$f(y)\leq \underset{z\in X}{\inf}f(z),\,\,\forall y\in \widetilde{X}\backslash X.$$
If $x\in X$ is a mountain pass point for $f_{|X}$ then $x$ is a mountain pass for $f:\widetilde{X}\rightarrow \mathbb{R}$.
\end{lem}

\begin{proof}
Let $x$ be a mountain pass point for $f_{|X}$. Let $\widetilde{\mathcal{N}}$ be an arbitrary neighborhood of $x$ in $\widetilde{X}$. By definition of induced topology we have that $\mathcal{N}:=\widetilde{\mathcal{N}}\cap X$ is a neighborhood of $x$ in $X$. Using the hypothesis we obtain the set equality
$$\mathcal{N}\cap \{y\in X\,|\, f (y) > f (x)\}=\widetilde{\mathcal{N}}\cap \{y\in \widetilde{X}\,|\, f (y) > f (x)\}.$$
This shows that $x$ is also a mountain pass point for $f:\widetilde{X}\rightarrow \mathbb{R}$.

\end{proof}


\begin{thebibliography}{99}
\bibitem{arsie-ebenbauer} {\bf Arsie A., Ebenbauer C.}, {\it Locating omega-limit sets using height functions}, J. Differential Equations, 248, 2458-2469 (2010).
\bibitem{bhatia-szego} {\bf Bhatia N.P., Szeg\"{o} G.P.}, {\it Stability Theory of Dynamical Systems}, Springer-Verlag, Berlin, 1970.
\bibitem{birtea-comanescu-geometric}{\bf P. Birtea, D. Com\u anescu}, {\it Geometric Dissipation for dynamical systems},
    Comm. Math. Phys., Vol. 316, Issue 2 (2012), pp. 375-394.
\bibitem{birtea-comanescu} {\bf Birtea P., Com\u{a}nescu D.}, {\it Invariant critical sets of conserved quantities}, Chaos, Solitons $\&$ Fractals, 44, 693-701 (2011).
\bibitem{comanescu-1}{\bf  D. Com\u{a}nescu}, {\it The stability problem for the torque-free gyrostat investigated by using algebraic methods}, Applied Mathematics Letters, Volume 25, Issue 9 (2012), pp. 1185-1190.
\bibitem{comanescu-2} {\bf  D. Com\u{a}nescu}, {\it Stability of equilibrium states in the Zhukovski case of heavy
gyrostat using algebraic methods}, Mathematical Methods in the Applied Sciences, Volume 36, Issue 4 (2013),pp. 373-382 .
\bibitem{comanescu-3} {\bf D. Com\u{a}nescu}, {\it A note on stability of the vertical uniform rotations of the heavy top}, ZAMM, Volume 93, Issue 9 (2013), pp. 697-699.
\bibitem{irtegov-titorenko} {\bf Irtegov V.D., Titorenko T.N.}, {\it The invariant manifolds of systems with first integrals}, J. Appl. Math. Mech., 73, 379-384 (2009).
\bibitem{jabri}{\bf Jabri Y.,}{\it The Mountain Pass Theorem. Variants, Generalizations and Some Applications.}, Cambridge Univ. Press, 2003. 
\bibitem{robinson} {\bf C. Robinson}, {\it Dynamical systems, Stability, Symbolic Dynamics, and Chaos}, CRC Press, 1995.
\bibitem{sternberg} {\bf S. Sternberg}, {\it Lectures on Differential Geoemetry}, Prentice Hall, Englewood Cliffs, NJ, 1964.

\end{thebibliography}
\end{document}